\documentclass[11pt]{article}
\usepackage[latin1]{inputenc}
\usepackage[margin=1in]{geometry} 
\usepackage{amsmath}
\usepackage{amssymb}
\usepackage{amsfonts}
\usepackage{amsthm}
\usepackage{nicefrac}
\usepackage{url}
\usepackage[ngerman,english]{babel}
\usepackage{comment}
\usepackage{algorithm}
\usepackage[noend]{algpseudocode}

\usepackage[dvipsnames]{xcolor}

\usepackage{tikz}
\usetikzlibrary{trees,arrows}
\usepackage{subfigure}
\usepackage{caption}
\usepackage{bbm}

\newcommand{\E}{\ensuremath{\mathbb{E}}}

\newcommand{\Prob}{\ensuremath{\mathbb{P}}}

\newcommand{\R}{\ensuremath{\mathbb{R}}}

\newtheorem{thm}{Theorem}[section]
\newtheorem{prop}[thm]{Proposition}
\newtheorem{coro}[thm]{Corollary}
\newtheorem{lem}[thm]{Lemma}

\begin{document}
\title{Probabilistic Analysis of the \\ Dual-Pivot Quicksort ``Count''}
\date{}
\author{Ralph Neininger and Jasmin Straub\\ Institute for Mathematics \\ Goethe University Frankfurt \\ Germany \\
\{neininger, jstraub\}{@}math.uni-frankfurt.de
}

\maketitle

\abstract{Recently, Aum\"uller and Dietzfelbinger proposed a version of a dual-pivot quicksort, called ``Count'', which is optimal among dual-pivot versions with respect to the average number of key comparisons required. In this note we provide further probabilistic analysis of ``Count''. We derive an exact formula for the average number of swaps needed by ``Count'' as well as an asymptotic formula for the variance of the number of swaps and a limit law. Also for the number of key comparisons the asymptotic variance and a limit law are identified. We also consider both complexity measures jointly and find their asymptotic correlation.}

\section{Introduction}\label{sec:Intro}
In 2009, Oracle replaced the sorting algorithm quicksort in its Java 7 runtime
library by a new dual-pivot quicksort due to Vladimir Yaroslavskiy. This was based on the good performance of the new dual-pivot version in experiments. Since then, various theoretical studies were devoted to explain, quantify, generalize and improve  the dual-pivot quicksort, starting with a rigorous average case analysis of Wild and Nebel \cite{WN12} for the numbers of comparisons and swaps in a uniform probabilistic model. However, classical cost parameters turned out to not be sufficient to explain the good performance of the new dual-pivot quicksort. Hence, authors also took the memory hierarchy of computer storage into account and studied cache effects, see Kushagra et al. \cite{KLQM14}, and quantities related to cache misses, in particular the number of scanned elements, see Nebel et al. \cite{WNM16} and \cite{ADK16}.

The basic idea to gain while using a dual-pivot quicksort is that during the partitioning stages elements may not need to be compared to both pivot elements. Thus, the theoretical question on how to optimally arrange the partitioning stages arose. Regarding the average number of key comparisons this has fully been solved in Aum\"uller et al. \cite{ADHKP16} where the dual-pivot quicksort version ``Count'' is identified to be optimal in this respect.

The dual-pivot quicksort ``Count'' uses two pivot elements which are chosen from the array as the first and the last element. Assume they are $p$ and $q$ with $p<q$. Then all other elements are compared to the pivot elements to classify them as smaller as $p$, between $p$ and $q$ (called medium) or being larger than $q$. Obviously, elements between $p$ and $q$ need to be compared to $p$ and $q$ to get classified. However, if an element smaller than $p$ is first compared to $p$, there is no need for also comparing it to $q$, the same for elements larger than $q$ if they are first compared to $q$. Now, to typically save some comparisons one keeps track on how many elements smaller than $p$ and larger than $q$ already have been identified. If you have classified more small than large elements so far, this indicates that the data are split unevenly by $p$ and $q$ with more small elements than large elements. Hence, in such a case, one compares the next element first to $p$, and vice versa. After partitioning the elements according to this rule the algorithm recurses on the lists of elements smaller than $p$, medium, and larger than $q$, see Algorithm~1.

The aim of the present note is to provide further probabilistic analysis of ``Count'' in the model of uniformly permuted (distinct) data. While there are little mathematical innovations required to study such quicksort variants there are still new  combinatorial structures to reveal and to study. Also, we feel that it is worth adding probabilistic analysis for ``Count'' due to its distinguished role. We study the numbers of comparisons and swaps required by ``Count'' with respect to expectations, variances, limit laws and correlations. Along our analysis also the number of scanned elements could be studied. In Section \ref{sec_res} our results are stated. Section \ref{sec:ana} contains the analysis and proofs of the theorems. At the end the version of ``Count'' used in the present note is stated explicitly (Algorithm~1) which differs slightly from the original version of Aum\"uller et al. \cite{ADHKP16}.

The results of this note are part of the second mentioned  author's master's thesis \cite{Straub}.

\section{Results} \label{sec_res}
Throughout, we assume that the data are distinct and uniformly permuted.
\subsection{Mean values}
The mean value for the number $\mathcal{C}_n$ of key comparisons has already been derived in Aum\"uller et al. \cite[Theorem 12.1]{ADHKP16}. For $n\geq4$, we have
\begin{align} 
	\mathbb{E}[\mathcal{C}_n] & =  \frac{9}{5}n\mathcal{H}_n-\frac{1}{5}n\mathcal{H}_n^{\text{alt}}-\frac{89}{25}n+\frac{67}{40}\mathcal{H}_n-\frac{3}{40}\mathcal{H}_n^{\text{alt}}-\frac{83}{800}+\frac{(-1)^n}{10} \nonumber\\
	&\quad\;\;~ -\frac{\mathbbm{1}_{\left\{n\; \mathrm{ even}\right\}}}{320}\left(\frac{1}{n-3}+\frac{3}{n-1}\right) +\frac{\mathbbm{1}_{\left\{n\; \mathrm{  odd}\right\}}}{320}\left(\frac{3}{n-2}+\frac{1}{n}\right)\nonumber\\
	&= \frac{9}{5}n\log(n)+A_cn+\frac{67}{40}\log(n)+\mathrm{O}(1)\quad (n\rightarrow\infty), \label{exp_mean_cn}
\end{align}
where $A_c=\frac{9}{5}\gamma+\frac{1}{5}\log(2)-\frac{89}{25}=-2.382...$,
\[ \mathcal{H}_n := \sum_{k=1}^{n}{\frac{1}{k}} \qquad \mbox{ and } \qquad \mathcal{H}_n^{\text{alt}}:=\sum_{k=1}^{n}{\frac{(-1)^k}{k}}. \]
Note, that, as $n\to\infty$, we have
\[ \mathcal{H}_n = \log(n) + \gamma + \mathrm{O}\left(\frac{1}{n}\right), \qquad \mathcal{H}_n^{\text{alt}}  = -\log(2)+\mathrm{O}\left(\frac{1}{n}\right),  \]
and $\gamma = 0.577...$ denotes the Euler--Mascheroni constant. \\
We add the corresponding formula for the mean number of swaps.

\begin{thm} \label{satz:swape} For the number of swaps $\mathcal{S}_n$ of the dual-pivot quicksort ``Count" in Algorithm~1  when sorting a random permutation of length $n$, for $n\geq4$, we have
\begin{align}\label{mean_swap} 
	\mathbb{E}[\mathcal{S}_n] &=  \frac{3}{4}n\mathcal{H}_n + \frac{1}{20}n\mathcal{H}_n^{\text{alt}} -\frac{4}{5}n+\frac{3}{4}\mathcal{H}_n+\frac{1}{20}\mathcal{H}_n^{\text{alt}} - \frac{23}{160}-\frac{(-1)^n}{40} \\
	&\quad\;\;~- \frac{\mathbbm{1}_{\left\{n\; \mathrm{ even}\right\}}}{320}\left(\frac{1}{n-3}+\frac{3}{n-1}\right) + \frac{\mathbbm{1}_{\left\{n\;\mathrm{ odd}\right\}}}{320}\left(\frac{3}{n-2}+\frac{1}{n}\right)\nonumber\\
	&= \frac{3}{4}n\log(n)+A_sn+\frac{3}{4}\log(n)+\mathrm{O}(1) \quad (n\rightarrow\infty),\nonumber
\end{align}
with $A_s=\frac{3}{4}\gamma-\frac{1}{20}\log(2)-\frac{4}{5}= -0.401...$
\end{thm}
The derivation of formula (\ref{mean_swap}) uses a surprising combinatorial identity  for which we have not yet found an intuitive explanation, see Proposition \ref{satz:einhalb}.

\subsection{Deviations and correlations}
For the orders of the standard deviations and correlations we have the following asymptotic results which are independent of the special implementation, i.e., valid for our version Algorithm~1  as well as for the version in Aum\"uller et al. \cite{ADHKP16}.
\begin{thm} \label{satz:dev_comp} For the numbers of key comparisons $\mathcal{C}_n$ and swaps $\mathcal{S}_n$ of the dual-pivot quicksort ``Count" when sorting a random permutation of length $n$, as $n\to\infty$, we have
\begin{align*}
	\mathrm{Var}(\mathcal{C}_n) &\sim \sigma_c^2 n^2,\\
	\mathrm{Var}(\mathcal{S}_n) &\sim \sigma_s^2 n^2,\\
	\mathrm{Cov}(\mathcal{C}_n,\mathcal{S}_n) &\sim \sigma_{c,s}^2 n^2,
\end{align*}
where
\begin{align*}
	\sigma_c^2 &=\frac{1609}{300}-\frac{27}{50}\pi^2+\frac{3}{10}\log 2  = 0.241\ldots,\\
	\sigma_s^2 &=\frac{47}{48}-\frac{3}{32}\pi^2+\frac{3}{32}\log 2  = 0.118\ldots,\\
	\sigma_{c,s}^2&= \frac{43}{20}-\frac{9}{40}\pi^2+\frac{7}{40}\log 2=0.050\ldots.
\end{align*}
Hence, the asymptotic correlation between $\mathcal{C}_n$ and $\mathcal{S}_n$ is
\begin{align*}
	\mathrm{Corr}(\mathcal{C}_n,\mathcal{S}_n) \sim 0.298755\ldots \quad (n\to\infty).
\end{align*}
\end{thm}
Note that for the classic quicksort (with one pivot element and the partitioning method of Hoare and Sedgewick, see, e.g., Wild \cite[Algorithm~2]{WILD}) we have a strong negative correlation of about $-0.864$ between the numbers of key comparisons and key exchanges, see \cite{Neininger}. The different nature of correlations for the classic quicksort and the dual-pivot quicksort ``Count'' is already transparent in the scatter plot in Figure 1.
\begin{figure}[ht]\begin{center}  \label{fig_1}
	\includegraphics[width=14cm]{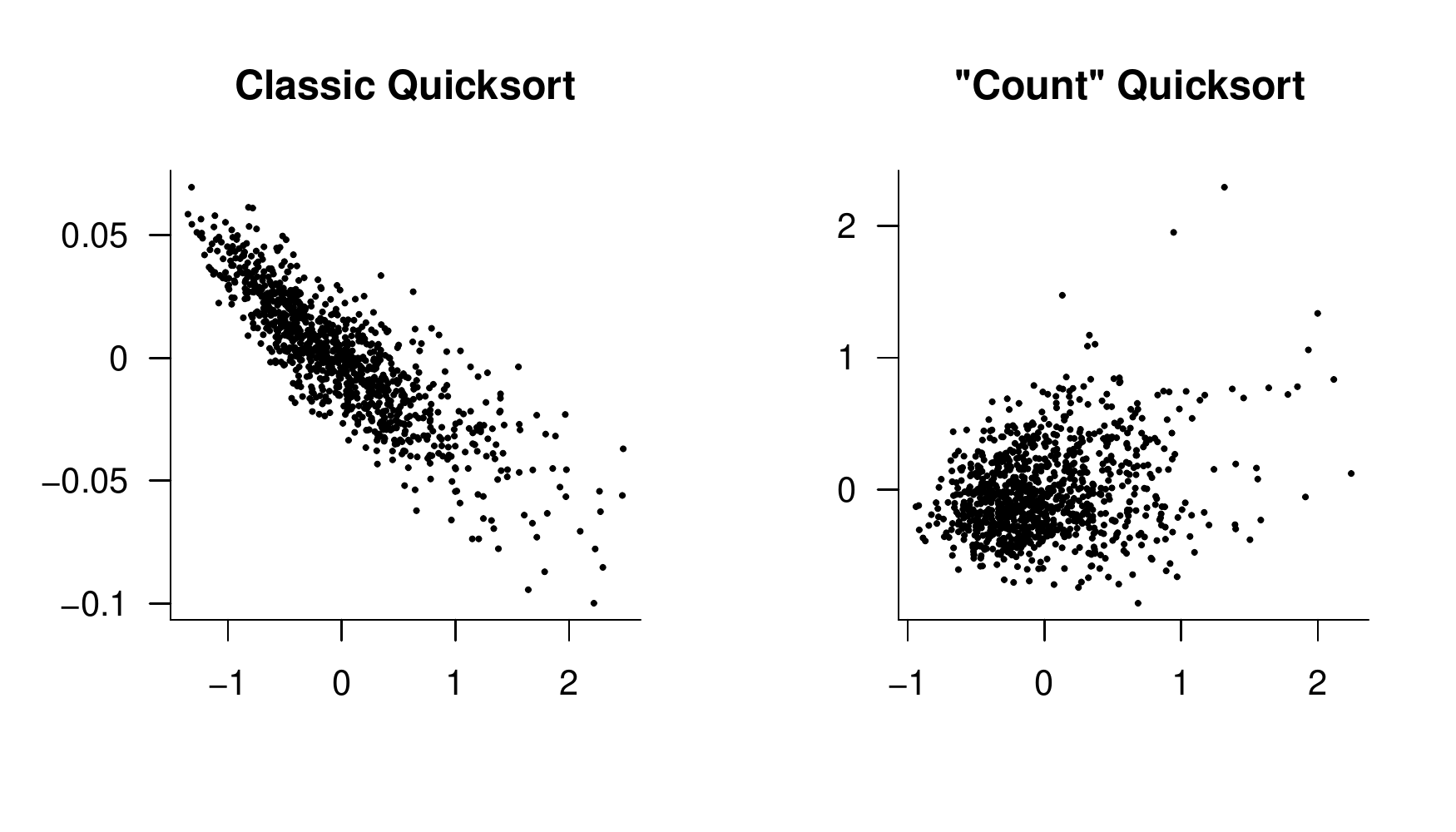}\end{center} \vspace{-1cm} 
	\caption{Scatter plots of the normalized versions of $\mathcal{C}_n$ and $\mathcal{S}_n$ as in Theorem \ref{limit_law} for the classic quicksort (left) and the dual-pivot quicksort ``Count'' (right). Depicted are $1000$ samples of random permutations of size $n=10000$.}
\end{figure}

\subsection{Limit distributions}
The asymptotic behaviour of the distributions of $\mathcal{C}_n$ and $\mathcal{S}_n$ is determined by the recursive distributional equation (RDE) for probability distributions on $\R^2$,
\begin{align}\label{bi-fpe}
	X \stackrel{d}{=} \sum_{r=1}^{3}{\left(\begin{matrix}D_r & 0 \\ 0 & D_r \end{matrix}\right) X^{(r)}}+\left(\!\! \begin{array}{c} b_1\\b_2\end{array}\!\!  \right),
\end{align}
with
\begin{align*}
	\left(\!\!  \begin{array}{c} b_1\\b_2\end{array}\!\!  \right)=\left( \!\! \begin{array}{c} 1+D_2+\min\left\{D_1,D_3\right\}+\frac{9}{5}\sum_{r=1}^{3}{D_r\log D_r}\\ D_1+D_3+\mathbbm{1}_{\left\{D_3>D_1\right\}}(\frac{1}{2}D_1+D_2-D_3)
	+\frac{3}{4}\sum_{r=1}^{3}{D_r\log D_r} \end{array}\!\!  \right),
\end{align*}
where $X^{(1)}$, $X^{(2)}$, $X^{(3)}$ and $(D_1,D_2,D_3)$ are independent, $X^{(r)}$ is distributed as $X$ for $r=1,2,3$ and $(D_1,D_2,D_3)$ is distributed as the spacings
\begin{align}\label{spacings}
	\left(\min\{U_1,U_2\},\max\{U_1,U_2\}-\min\{U_1,U_2\},
	1-\max\{U_1,U_2\}\right)
\end{align}
of two independent, unif$[0,1]$ distributed random variables $U_1$ and $U_2$ in $[0,1]$. Among all centered distributions with a finite second moment RDE (\ref{bi-fpe}) has a unique solution, see \cite[Lemma 3.1]{Neininger}. We denote this solution by
\begin{align}\label{def_bi_limit}
	{\cal L}(X)={\cal L}\left(\!\!  \begin{array}{c} X_c\\X_s\end{array}\!\!  \right).
\end{align}
Note that the first and second coordinate of RDE (\ref{bi-fpe}) give univariate RDEs which also characterize the marginals ${\cal L}(X_c)$ and ${\cal L}(X_s)$, respectively.
\begin{thm}\label{limit_law}
	For the normalized numbers of key comparisons and swaps of the dual-pivot quicksort ``Count'' when sorting a random permutation of length $n$ we have jointly in distribution
\begin{align*}
	\left(\frac{\mathcal{C}_n-\E[\mathcal{C}_n]}{n}, \frac{\mathcal{S}_n-\E[\mathcal{S}_n]}{n}\right)\stackrel{d}{\longrightarrow}  (X_c,X_s),\quad (n\to\infty),
\end{align*}
where $(X_c,X_s)$ is given in (\ref{def_bi_limit}).
\end{thm}
From the form of (\ref{bi-fpe}) we obtain for the limits $X_c$ and $X_s$ the existence of densities.
\begin{coro}\label{coro_dens}
The distributions of $X_c$ and $X_s$ have densities each being infinitely smooth and rapidly decreasing (together with all their higher derivatives).
\end{coro}

\section{Analysis}\label{sec:ana}
In this section, we sketch the analysis leading to our results. Recall that we assume that the array is uniformly distributed. For simplicity, we assume that the input consists of a sequence $(U_i)_{i\ge 1}$ of independent, uniformly on $[0,1]$ distributed random variables and that, for fixed $n$, we have to sort the array $[U_1,\ldots,U_n]$. To further simplify the analysis we instead of choosing $U_1$ and $U_n$ as pivot elements, cf.~Algorithm~1, choose $U_1$ and $U_2$ as pivot elements. (This does not change the distributions of our quantities but implies that various convergences also hold in a strong (almost sure, $L_2$) sense.)  

We start with the growth of the numbers of elements smaller than $p$, between $p$ and $q$, and larger than $q$ during the partitioning stage. We denote by $(S_i,M_i,L_i)$ the vector of these numbers after $i$ elements have been compared to the pivot elements, $i=0,\ldots,n-2$. We first analyse an urn model which captures the growth dynamics of $(S_i,M_i,L_i)_{i=0,\ldots,n-2}$.\\

\noindent
{\bf A related urn model.} Consider a P\'olya--Eggenberger urn model with three types $\frak{s}$, $\frak{m}$ and $\frak{l}$ of balls (corresponding to ``small'', ``medium'' and ``large''  elements). Initially, the urn contains one ball of each type. In each step a ball is drawn uniformly from the urn (and independently from the earlier draws) and returned to the urn together with one ball of the same type, i.e., the replacement matrix is the $3\times 3$ identity matrix. Denote by $(S'_i,M'_i,L'_i)$ the vector of numbers of balls of each type added during steps $j=1,\ldots,i$. In other words, the composition of the urn after $i$ steps consists of $S'_i+1$ elements of type $\frak{s}$, of $M'_i+1$ elements of type $\frak{m}$, and of $L'_i+1$ elements of type $\frak{l}$. It is easy to see and has been used earlier, see \cite{ADHKP16} and \cite{WILD}, that for each $n\ge 3$ the processes $(S'_i,M'_i,L'_i)_{i=0,\ldots,n-2}$ and $(S_i,M_i,L_i)_{i=0,\ldots,n-2}$ are equal in distribution. Since we are only analysing parameters which are identified by distributions we hence drop the primes and shortly write $(S_i,M_i,L_i)$ for $(S'_i,M'_i,L'_i)$. We need the following conditional probability of adding a ball of type $\frak{l}$ while there are more balls of type $\frak{l}$ than of type $\frak{s}$.
\begin{prop}\label{satz:einhalb} For all $i\ge 1$ we have
\begin{align*}
	\Prob(L_{i+1}=L_i+1\,|\,L_i>S_i)=\frac{1}{2}.
\end{align*}
\end{prop}
\begin{proof}
It is easy to see that for each $i\ge 0$ the vector $(S_i,M_i,L_i)$ is uniformly distributed over its possible values, i.e., for all $s, m, \ell\ge 0$ with $s+m+\ell=i$, we have
\begin{align}\label{comb_1}
	\Prob((S_i, M_i, L_i)=(s,m,\ell)) = \frac{2}{(i+1)(i+2)}.
\end{align}
This implies that $\Prob(L_i=S_i)$ is $1/(i+1)$ if $i$ is even and $1/(i+2)$ if $i$ is odd. Hence, by symmetry, we obtain
\begin{align}\label{comb_2}
	\Prob(L_i>S_i) = \begin{cases} \frac{i}{2(i+1)},\quad &\text{ if $i$ is even,} \\ \frac{i+1}{2(i+2)},\quad &\text{ if $i$ is odd.}\end{cases}
\end{align}
On the other hand, by the law of total probability and (\ref{comb_1}), we have
\begin{align}
	\lefteqn{\Prob(L_{i+1}=L_i+1, L_i>S_i)}\nonumber\\
	&= \sum_{\ell=1}^{i}\sum_{\substack{s=0\\s\leq \ell-1}}^{i-\ell}{\mathbb{P}(S_i=s, M_i=i-s-\ell, L_i=\ell)\ \mathbb{P}(L_{i+1}=L_i+1\, |\, S_i=s, M_i=i-s-\ell, L_i=\ell)}\nonumber \\
	& = \sum_{\ell=1}^{i}\sum_{\substack{s=0\\s\leq \ell-1}}^{i-\ell}{\frac{2}{(i+1)(i+2)}\cdot\frac{\ell+1}{i+3}}\nonumber \\
	& = \begin{cases} \frac{i}{4(i+1)} ,\quad &\text{ if $i$ is even,} \\
                  \frac{i+1}{4(i+2)} ,\quad &\text{ if $i$ is odd.}\end{cases}\label{comb_3}
\end{align}
Combining (\ref{comb_2}) and (\ref{comb_3}) implies the assertion.
\end{proof}

\noindent {\bf Expected values.}
To study the expectation of the number of swaps $\mathcal{S}_n$ denote, for $n\ge 2$, by $S_n^+$, $M_n^+$ and $L_n^+$  the number of small, medium and large elements compared to $q$ first when partitioning an input sequence of length $n$ with the dual-pivot quicksort ``Count". Formally, we have
\begin{align*}
	S_n^+ &= \sum_{i=1}^{n-3}{\mathbbm{1}_{\left\{L_i>S_i,\ S_{i+1}=S_i+1\right\}}}, \\
	M_n^+ &= \sum_{i=1}^{n-3}{\mathbbm{1}_{\left\{L_i>S_i,\ M_{i+1}=M_i+1\right\}}},\\
	L_n^+ &= \sum_{i=1}^{n-3}{\mathbbm{1}_{\left\{L_i>S_i,\ L_{i+1}=L_i+1\right\}}}.
\end{align*}
Furthermore, we denote by $I^{(n)}=(I^{(n)}_1,I^{(n)}_2,I^{(n)}_3)$ the sizes of the three sublists generated in the first partitioning stage. Note, that we have $\frac{1}{n}I^{(n)}\to (D_1,D_2,D_3)$ in $L_p$ for any $1\le p<\infty$ as $n\to\infty$ (the $U_1$ and $U_2$ in (\ref{spacings}) are now the pivot elements).

The number $T_S(n)$ of swaps executed by the dual-pivot ``Count'' during the first partitioning stage is composed by
\begin{itemize}
\item[$\triangleright$] $I_1^{(n)}+\frac{1}{2}S_n^+$ swaps for the small elements (since there is a swap for each small element compared to $p$ first and a \emph{rotate3}-operation, i.e., $\nicefrac{3}{2}$ swaps, for each small element compared to $q$ first),
\item[$\triangleright$] $I_3^{(n)}-L_n^+$  swaps for the large elements compared to $p$ first,
\item[$\triangleright$] $M_n^+$ swaps for the medium elements compared to $q$ first, and
\item[$\triangleright$] two swaps at the end in order to bring the pivots to their final positions (both, line 33 and line 34 of Algorithm~1 need two write-accesses to the array).
\end{itemize}
Hence, we have
\begin{align}
	T_S(n)=2+I_1^{(n)}+I_3^{(n)}+\frac{1}{2}S_n^++ M_n^+ -L_n^+. \label{tolls}
\end{align}
We first derive the expectation of $T_S(n)$, see (\ref{tsn}).
\begin{lem} For all $n \geq 2$ we have
\begin{align*}
	\mathbb{E}[S_n^+] = \mathbb{E}[L_n^+-M_n^+]= \frac{1}{12}n-\frac{7}{24}+\frac{1}{8(n-\mathbbm{1}_{\left\{n\;\mathrm{ even}\right\}})}.
\end{align*}
\label{lemma:spluse}
\end{lem}
\begin{proof}
Similarly to the derivation of (\ref{comb_3}) we find
\begin{align}
	\Prob(L_i>S_i, S_{i+1}=S_i+1)& = \sum_{\ell=1}^{i}\sum_{\substack{s=0\\s\leq \ell-1}}^{i-\ell}{\frac{2}{(i+1)(i+2)}\cdot\frac{s+1}{i+3}}\nonumber \\
	 &= \begin{cases} \frac{i(i+4)}{12(i+1)(i+3)} ,\quad &\text{ if $i$ is even,} \\
                  \frac{1}{12} ,\quad &\text{ if $i$ is odd.} \end{cases}\label{comb_4}
\end{align}
By definition,
\begin{align*}
	\E[S_n^+]
  = \mathbb{E}\left[\sum_{i=1}^{n-3}{\mathbbm{1}_{\left\{L_i>S_i,\ S_{i+1}=S_i+1\right\}}}\right]= \sum_{i=1}^{n-3}\mathbb{P}(L_i>S_i,\ S_{i+1}=S_i+1).
\end{align*}
Now, we plug  in expression (\ref{comb_4}) on the right hand side of the latter display and use induction. This implies the claimed expression for $\mathbb{E}[S_n^+]$. Furthermore, using Proposition \ref{satz:einhalb} we have
\begin{align*} 
	\lefteqn{\E[S_n^++M_n^+-L_n^+]}\\
	& = \sum_{i=1}^{n-3}{\mathbb{P}(L_i>S_i,\ S_{i+1}=S_i+1)+\mathbb{P}(L_i>S_i,\ M_{i+1}=M_i+1)-\mathbb{P}(L_i>S_i,\ L_{i+1}=L_i+1)} \\
	& = \sum_{i=1}^{n-3}{\mathbb{P}(L_i>S_i)\Big(\mathbb{P}(L_{i+1}=L_i\,|\,L_i>S_i)
	-\mathbb{P}(L_{i+1}=L_i+1\,|\,L_i>S_i)\Big)} \\
	& = 0,
\end{align*}
hence $\mathbb{E}[S_n^+] = \mathbb{E}[L_n^+-M_n^+]$.
\end{proof}
Since by symmetry we have $\E[I^{(n)}_r]=(n-2)/3$ for $r=1,2,3$ we obtain from (\ref{tolls}) and Lemma \ref{lemma:spluse} that
\begin{align} 
	\mathbb{E}[T_S(n)] =\frac{5}{8}n+\frac{13}{16}- \frac{1}{16(n-\mathbbm{1}_{\left\{n\text{ even}\right\}})}. \label{tsn}
\end{align}
We are now prepared to prove Theorem \ref{satz:swape}.
\begin{proof} {\em (Theorem \ref{satz:swape})} We draw back to earlier analysis of dual-pivot quantities which led to similar recurrences for their expectations with other toll functions, see Wild \cite[Section 4.2.1]{WILDM}. We have the recurrence
\begin{align*} 
	\mathbb{E}[\mathcal{S}_n] = \frac{6}{n(n-1)} \sum_{k=0}^{n-2}{(n-k-1)\mathbb{E}[\mathcal{S}_k]}+\mathbb{E}[T_S(n)],\quad (n\ge 2)
\end{align*}
with $\E[\mathcal{S}_0]=\E[\mathcal{S}_1]=0$. From Wild's work we know that
\begin{align} 
	\mathbb{E}[\mathcal{S}_n] &=  \frac{1}{\binom{n}{4}} \sum_{i=5}^{n}{\binom{i}{4}\sum_{j=3}^{i-2}{\left(\mathbb{E}[T_S(j+2)]
	-\frac{2j}{j+2}\mathbb{E}[T_S(j+1)]+\frac{j(j-1)}{(j+2)(j+1)}
	\mathbb{E}[T_S(j)]\right)}} \notag \\
	&\quad~ +\frac{n+1}{5}\left(\mathbb{E}[T_S(4)]+\frac{1}{2}\mathbb{E}[T_S(2)]\right),\quad (n\ge 4). \label{closedform}
\end{align}
Now, the affine component $\frac{5}{8}n+\frac{13}{16}$ of  $\mathbb{E}[T_S(n)]$ in (\ref{tsn}) implies an overall contribution to $\mathbb{E}[\mathcal{S}_n]$, cf.~\cite[Section 4.2.1.1]{WILDM}, of
\begin{align} 
	\frac{3}{4}n\mathcal{H}_n + \frac{3}{4}\mathcal{H}_n - \frac{33}{40}n-\frac{27}{160}. \label{teil3}
\end{align}
To identify the contribution of summand $(16(n-\mathbbm{1}_{\left\{n\text{ even}\right\}}))^{-1}$ of $\mathbb{E}[T_S(n)]$ in (\ref{tsn}) to $\mathbb{E}[\mathcal{S}_n]$ we set $\mu_n:=(n-\mathbbm{1}_{\left\{n\text{ even}\right\}})^{-1}$. In view of (\ref{closedform}) we need to compute
\begin{align}\label{doub_sum}
	\frac{1}{\binom{n}{4}} \sum_{i=5}^{n}{\binom{i}{4}\sum_{j=3}^{i-2}{\left(\mu_{j+2}-\frac{2j}{j+2}\mu_{j+1}
	+\frac{j(j-1)}{(j+2)(j+1)}\mu_j\right)}}.
\end{align}
For the inner sum in (\ref{doub_sum}) we find, by plugging in the values of $\mu_j$, $\mu_{j+1}$ and $\mu_{j+2}$ and using induction, for $i\ge 5$, that
\begin{align*}
	\sum_{j=3}^{i-2}{\left(\mu_{j+2}-\frac{2j}{j+2}\mu_{j+1}+\frac{j(j-1)}{(j+2)(j+1)}
	\mu_j\right)} = -4\mathcal{H}_i^{\text{alt}}+(-1)^i\frac{2}{i}-\frac{17}{6}.
\end{align*}
Plugging this expression into (\ref{doub_sum}) and another induction imply, for $n\ge 4$, that
\begin{align*}
	\lefteqn{\frac{1}{\binom{n}{4}}\sum_{i=5}^{n}{\binom{i}{4}\left(-4\mathcal{H}_i^{\text{alt}}+\frac{2(-1)^i}{i} -\frac{17}{6}\right)}}\\
	& =  -\frac{17}{30}n-\frac{4}{5}n\mathcal{H}_n^{\text{alt}}-\frac{4}{5}\mathcal{H}_n^{\text{alt}}-\frac{17}{30}+\frac{2}{5}(-1)^n \\
	&\quad~ + \frac{\mathbbm{1}_{\left\{n\;\text{even}\right\}}}{20}
	\left(\frac{1}{n-3}+\frac{3}{n-1}\right) - \frac{\mathbbm{1}_{\left\{n\,\text{odd}\right\}}}{20}
	\left(\frac{3}{n-2}+\frac{1}{n}\right).
\end{align*}
Hence, the total contribution of  summand $(16(n-\mathbbm{1}_{\left\{n\text{ even}\right\}}))^{-1}$ of $\mathbb{E}[T_S(n)]$ to $\mathbb{E}[\mathcal{S}_n]$ is $\frac{1}{16}$ times
\begin{align*} 
	-\frac{2}{5}n-\frac{4}{5}(n+1)\mathcal{H}_n^{\text{alt}}-\frac{2}{5}
	+\frac{2}{5}(-1)^n
	+ \frac{\mathbbm{1}_{\left\{n\;\text{even}\right\}}}{20}
	\left(\frac{1}{n-3}+\frac{3}{n-1}\right) - \frac{\mathbbm{1}_{\left\{n\;\text{odd}\right\}}}{20}
	\left(\frac{3}{n-2}+\frac{1}{n}\right).
\end{align*}
Together with (\ref{teil3}) this implies the assertion.
\end{proof}
An alternative route to the formula for $\E[\mathcal{S}_n]$ in  Theorem \ref{satz:swape} is via generating functions as used in \cite{ADHKP16} to derive $\E[\mathcal{C}_n]$, see \cite{Straub} for details.\\

\noindent
{\bf Variances, correlation and limit laws.} The results stated in Theorems \ref{satz:dev_comp} and \ref{limit_law} follow from a standard application of the contraction method based on the expansions of $\E[\mathcal{C}_n]$ and $\E[\mathcal{S}_n]$  in (\ref{exp_mean_cn}) and in Theorem \ref{satz:swape}, respectively. We refer the reader to Section 4 of \cite{WNN15} where the application of the contraction method to dual-pivot quicksort is explained and worked out for  parameters with other toll functions. To apply this approach to $\mathcal{C}_n$ and   $\mathcal{S}_n$ we only need to derive the asymptotic $L_2$-behaviour of our toll functions in the recurrences for $\mathcal{C}_n$ and $\mathcal{S}_n$.

Note, that the number of key comparisons executed during the first partitioning stage of the dual-pivot quicksort ``Count'' is given by
\[ T_C(n) = n-1+I_2^{(n)}+I_3^{(n)}+S_n^+-L_n^+. \]
For the normalized quantities $X_0:=(0,0)^t$ and
\begin{align}\label{normalized}
	X_n:=\left(\frac{\mathcal{C}_n-\mathbb{E}[\mathcal{C}_n]}{n},\frac{\mathcal{S}_n-\mathbb{E}[\mathcal{S}_n]}{n}\right)^t,\quad n\ge 1,
\end{align}
we have the distributional recurrence
\begin{align}
	X_n \stackrel{d}{=} \sum_{r=1}^{3}{A_r^{(n)}X_{I_r^{(n)}}^{(r)}}+b^{(n)},\quad n\ge 2,
\end{align}
where $(X_j^{(1)})_{0\leq j\leq n}$, $(X_j^{(2)})_{0\leq j\leq n}$, $(X_j^{(3)})_{0\leq j\leq n}$ and $(b^{(n)},I^{(n)})$ are independent, $X_j^{(r)}$ is distributed as $X_j$ for $r=1,2,3$, $j\geq0$ and
\[ A_r^{(n)}=\frac{1}{n}\left(\begin{matrix}I_r^{(n)} & 0 \\ 0 & I_r^{(n)}\end{matrix}\right) \quad \text{and} \quad b^{(n)}=\frac{1}{n}\left( \begin{array}{c}T_C(n)-\mathbb{E}[\mathcal{C}_n]+\sum\limits_{r=1}^{3}{\mathbb{E}[\mathcal{C}_{I_r^{(n)}}|I_r^{(n)}]}\\T_S(n)-\mathbb{E}[\mathcal{S}_n]+\sum\limits_{r=1}^{3}{\mathbb{E}[\mathcal{S}_{I_r^{(n)}}|I_r^{(n)}]}\end{array} \right). \]
Recall that we have the $L_2$-convergence $A_r^{(n)}\to \mathrm{diag}(D_r,D_r)$ for $r=1,2,3$. A standard calculation based on the expansions of $\E[\mathcal{C}_n]$ and $\E[\mathcal{S}_n]$ implies that
\begin{align*}
	\frac{1}{n}\left( \begin{array}{c}-\mathbb{E}[\mathcal{C}_n]+\sum\limits_{r=1}^{3}
	{\mathbb{E}[\mathcal{C}_{I_r^{(n)}}|I_r^{(n)}]}\\
	-\mathbb{E}[\mathcal{S}_n]+\sum\limits_{r=1}^{3}{\mathbb{E}[\mathcal{S}_{I_r^{(n)}}|I_r^{(n)}]}\end{array} \right)\stackrel{L_2}{\longrightarrow}\left( \begin{array}{c} \frac{9}{5}\sum\limits_{r=1}^3 D_r \log D_r \\  \frac{3}{4}\sum\limits_{r=1}^3 D_r \log D_r \end{array} \right).
\end{align*}
Hence, it remains to identify the limits of $T_C(n)/n$  and $T_S(n)/n$ which reduces to the limit of $\frac{1}{n}(S_n^+,M_n^+,L_n^+)$.
\begin{lem}\label{lemma:SnMnLn} As $n\rightarrow\infty$, we have
\begin{align} 
	\frac{1}{n}(S_n^+,M_n^+,L_n^+) \stackrel{L_2}{\longrightarrow} \left( \mathbbm{1}_{\left\{D_3>D_1\right\}}D_1, \mathbbm{1}_{\left\{D_3>D_1\right\}}D_2, \mathbbm{1}_{\left\{D_3>D_1\right\}}D_3\right). 
\end{align}
\end{lem}
\begin{proof} We consider the random lattice path $W=(W_i)_{i\geq 0}$ defined by $W_0:=0$ and, for $i\geq1$,
\[ W_i = \sum_{j=1}^{i}\left(\mathbbm{1}_{\left\{\text{the $j$-th classified element is large}\right\}}-\mathbbm{1}_{\left\{\text{the $j$-th classified element is small}\right\}}\right),\] i.e., $W_i$ is the difference of the number of large and small elements after having classified $i$ elements. Conditionally on  $(D_1,D_2,D_3)=(d_1,d_2,d_3)$, the process $W$ is a simple random walk on $\mathbb{Z}$ going one step down with probability $d_1$, staying at its current state with probability $d_2$ and  going one step up with probability $d_3$.

We first consider the case $d_1<d_3$. Here, the random walk $(W_i)$ has a positive drift. From the strong law of large numbers, we obtain that $W_i$ tends to $+\infty$ almost surely. This implies that on $\left\{D_3>D_1\right\}$, there exists almost surely some random $n_0\in\mathbb{N}$ such that $W_i>0$ for all $i\geq n_0$. This means that  from index $n_0$ on we always compare to pivot $q$ first. Hence, on $\left\{D_3>D_1\right\}$ we obtain
\[ \frac{I_1^{(n)}-n_0}{n} \leq \frac{S_n^+}{n} \leq \frac{I_1^{(n)}}{n}. \]
Thus, on $\left\{D_3>D_1\right\}$, we have that  $S_n^+/n$ converges to $D_1$ almost surely.

Similarly, if $d_1>d_3$ the random walk $W$ has a negative drift and we find   $S_n^+/n \to 0$ almost surely on $\left\{D_3<D_1\right\}$.

Overall, using dominated convergence, we find $$\frac{S_n^+}{n}\stackrel{L_2}{\longrightarrow} \mathbbm{1}_{\left\{D_3>D_1\right\}}D_1.$$ The convergences of the other two components in the formulation of the present lemma follow similarly.
\end{proof}
Based on the $L_2$-convergences derived in the present subsection the results of Theorems  \ref{satz:dev_comp} and \ref{limit_law} follow along the lines of Section 4 of \cite{WNN15}.

Corollary \ref{coro_dens} follows from a general theorem of Leckey \cite{Leckey} based on techniques of Fill and Janson \cite{FiJa00}.\\

\noindent
{\bf Acknowledgement:} We thank Sebastian Wild for ongoing advice, in particular to count a {\em rotate3}-operation as $\nicefrac{3}{2}$ of a swap.\\ \vspace{10cm}

\begin{algorithm}
\caption{Dual-Pivot Quicksort Algorithm ``Count". Slight modifications to the version of Aum\"uller et al.~\cite[Algorithm~1]{ADHKP16} are marked and commented on in blue color.}
\begin{algorithmic}[1]\label{alg_count}
\Procedure{Count}{$A$, $\mathit{left}$, $\mathit{right}$}
\If {$\mathit{right} \leq \mathit{left}$} \State return \EndIf
\If {$A[\mathit{right}]<A[\mathit{left}]$} \color[rgb]{0.1,0.1,0.9}  \State $p \gets A[\mathit{right}]$; $q \gets A[\mathit{left}]$ \quad \footnotesize // in \cite{ADHKP16}: swap $A[\mathit{left}]$ and $A[\mathit{right}]$ \normalsize
\Else \State $p \gets A[\mathit{left}]$; $q \gets A[\mathit{right}]$ \quad \footnotesize // in \cite{ADHKP16}: $p\gets A[\mathit{left}]$; $q\gets A[\mathit{right}]$ instead of lines 6 and 7  \normalsize \EndIf
\color[rgb]{0,0,0} \State $i \gets \mathit{left}+1$; $k \gets \mathit{right}-1$; $j \gets i$
\State $d \gets 0$
\While{$j\leq k$}
  \If{$d\geq0$}
	  \If{$A[j]<p$}
		  \State swap $A[i]$ and $A[j]$
			\State $i\gets i+1$; $j\gets j+1$; $d\gets d+1$
		\Else
		  \If{$A[j]<q$}
			  \State $j\gets j+1$
			\Else
			  \State swap $A[j]$ and $A[k]$
				\State $k\gets k-1$; $d\gets d-1$
			\EndIf
		\EndIf
	\Else
	  \If{$A[k]>q$}
		  \State $k\gets k-1$; $d\gets d-1$
		\Else
		  \If{$A[k]<p$}
			  \State // Perform a cyclic rotation to the left, i.e.,
				\State // tmp $\gets A[k]$; $A[k]$ $\gets$ $A[j]$; $A[j]$ $\gets$ $A[i]$; $A[i] \gets$ tmp
			  \State \emph{rotate3}($A[k],A[j],A[i]$)  
				\State $i\gets i+1$; $d\gets d+1$
			\Else
			  \State swap $A[j]$ and $A[k]$
			\EndIf
			\State $j\gets j+1$
		\EndIf
	\EndIf
\EndWhile
\color[rgb]{0.1,0.1,0.9} \State $A[\mathit{left}] \gets A[i-1]$ and $A[i-1] \gets p$ \quad\quad \footnotesize // in \cite{ADHKP16}: swap $A[\mathit{left}]$ and $A[i-1]$\normalsize
\State $A[\mathit{right}] \gets A[k+1]$ and $A[k+1] \gets q$ \quad  \footnotesize// in \cite{ADHKP16}: swap $A[\mathit{right}]$ and $A[k+1]$ \normalsize
\color[rgb]{0,0,0} \State \textsc{Count}($A$, $\mathit{left}$, $i-2$)
\State \textsc{Count}($A$, $i$, $k$)
\State \textsc{Count}($A$, $k+2$, $\mathit{right}$)
\EndProcedure
\end{algorithmic}
\end{algorithm}

\end{document}